\documentclass[preprints,article,submit,pdftex,moreauthors]{Definitions/mdpi} 
\firstpage{1} 
\pubvolume{1}
\issuenum{1}
\articlenumber{0}
\pubyear{2026}
\copyrightyear{2026}
\datereceived{ } 
\daterevised{ } 
\dateaccepted{ } 
\datepublished{ } 

\usepackage{etoolbox}
\preto{\abstractkeywords}{\nolinenumbers}

\Title{Finite-Resolution Information from Collision Statistics}

\Author{Alexander J. Gates $^{1}$\orcidA{}*}

\AuthorNames{Alexander J. Gates}

\address{%
$^{1}$ \quad School of Data Science, University of Virginia}

\corres{Correspondence: agates@virginia.edu}

\abstract{
Collision statistics provide a finite-resolution view of information by measuring how often a fixed number of independent samples fall on the same state.
These directly countable quantities form the basis of integer-order R\'enyi entropies. 
Here, we use low-order R\'enyi entropies to approximate Shannon entropy and mutual information, while characterizing what is necessarily lost when only finitely many collision moments are used.
We derive interpolation-error bounds showing that approximation error is controlled by the shape of the R\'enyi entropy path near the Shannon point. 
We also separate this deterministic error from finite-sample estimation error: for fixed collision order, increasing sample size improves estimation of the finite-resolution target but does not eliminate its deterministic difference from Shannon entropy or mutual information.
Finally, we show that finite collision moments do not generally identify Shannon entropy, and that increasing collision order shifts sensitivity toward high-probability events. 
Numerical experiments illustrate the approximation--estimation tradeoff and compare collision-based approximations with plug-in and Miller--Madow estimators. 
The framework links collision counts, R\'enyi entropy, Shannon limits, and mutual information through a finite-resolution view of information, clarifying when low-order coincidence structure is informative and when irreducible information is lost.
}

\begin{document}

\nolinenumbers

\section{Introduction}

Estimating Shannon entropy and mutual information from finite discrete samples remains a persistent problem, especially when the alphabet is large relative to the number of observations. 
The empirical plug-in estimator of entropy is biased downward, while the plug-in estimator of mutual information is biased upward, with both effects becoming pronounced in sparse regimes \cite{miller1955note, paninski2003estimation, paninski2004estimation}. 
Many estimators have been developed to address this problem, including Miller--Madow-type corrections \cite{miller1955note}, finite-sample bias corrections \cite{grassberger1988finite}, Bayesian estimators such as NSB \cite{nemenman2002entropy}, and minimax-optimal estimators based on approximation theory \cite{valiant2011estimating, jiao2015minimax, wu2016minimax}. 
These methods target Shannon entropy or mutual information directly.

Here, we take a complementary approach motivated by estimability. 
Rather than proposing another estimator of the Shannon functional, we formalize a family of finite-resolution information targets built from low-order collision moments. 
A collision moment measures the probability that two, three, or more independent samples take the same value. 
These probabilities are directly countable from data, have unbiased fixed-order estimators as $U$-statistics \cite{hoeffding1948class, serfling1980approximation}, and form the power sums underlying integer-order R\'enyi entropies \cite{renyi1961measures}. 
Related work has studied the estimation of R\'enyi entropies from samples, including the behavior of collision-based estimators for fixed entropy orders \cite{skorski2023renyi}. 
Our focus is different: we use these finite-order quantities as observations of the R\'enyi entropy path in order to construct and analyze finite-resolution approximations to Shannon entropy and mutual information.
Collision moments are also familiar in diversity measurement: the second-order collision probability is Simpson concentration, its complement is Simpson diversity, and its reciprocal is the order-two Hill number or effective number of types \cite{simpson1949measurement, hill1973diversity}. 
Thus collision statistics sit at the intersection of entropy estimation, R\'enyi information, and count-based diversity measurement.

This perspective leads to a different question from classical Shannon estimation. 
Instead of asking how to recover \(H(P)\) or \(I(X;Y)\) as accurately as possible from finite samples, we ask: \emph{how much of Shannon entropy is visible from finitely many low-order collision moments?}
This distinction matters because two sources of error are often conflated. 
The first is statistical estimation error: how accurately can a collision moment be estimated from data? 
The second is approximation error: how much information is lost when the population target itself is restricted to finitely many collision moments? 
For fixed collision order, increasing the sample size reduces the first error but not the second. 
The order of the collision moments determines the finite-resolution population target, while the sample size determines how accurately that target is estimated.

The connection to R\'enyi entropy provides the organizing structure for this finite-resolution view. 
For a distribution $P=(p_1,\ldots,p_m)$, the order-$k$ collision probability
\[
C_k(P)=\sum_i p_i^k
\]
defines the integer-order R\'enyi entropy
\[
H_k(P)=\frac{1}{1-k}\log C_k(P),
\]
while Shannon entropy is recovered by the continuous limit $k\to 1$ \cite{renyi1961measures, cover2006elements}. 
We therefore view the sequence $H_2(P),H_3(P),\ldots$ as a set of finite-resolution observations of the R\'enyi entropy path away from the Shannon point. 
Low-order extrapolations from this path yield approximations to Shannon entropy, and applying the same construction to the marginal and joint distributions yields corresponding approximations to mutual information.

The novelty of this paper is not a new minimax estimator, nor a claim that low-order collision statistics should replace direct Shannon estimators. 
Instead, the contribution is to make explicit the population targets defined by finite collision information and to separate their estimation error from their deterministic approximation error. 
For fixed order, a collision-based estimator converges to a finite-resolution R\'enyi-based functional, not generally to Shannon entropy or mutual information. 
Changing the collision orders changes the target itself. 
This makes collision approximations useful as interpretable finite-resolution summaries: they reveal how much information is captured by coincidence structure at a given order, and when irreducible information remains outside that finite moment representation.

The paper makes four contributions. 
First, we formulate low-order collision statistics as finite-resolution observations of the R\'enyi entropy path and use them to construct extrapolated approximations to Shannon entropy. 
Second, we extend the same construction to mutual information by combining marginal and joint entropy approximations. 
Third, we separate estimation error from approximation error, showing that fixed-order collision estimators converge to finite-resolution population targets rather than to Shannon quantities themselves. 
Fourth, we show that finite collision-moment truncations do not generally identify entropy, establishing an intrinsic limitation of any representation based only on finitely many low-order collisions. 
Numerical experiments illustrate these approximation--estimation tradeoffs and compare the resulting finite-resolution approximations with plug-in and Miller--Madow estimators.

The remainder of the paper is organized as follows. 
Section~\ref{sec:preliminaries} introduces collision probabilities, R\'enyi entropy, and the corresponding entropy contrast for mutual information.
Section~\ref{sec:low-order-approximations} develops low-order approximations to Shannon entropy and mutual information.
Section~\ref{sec:fixed-order-asymptotics} separates estimation error from approximation error. 
Section~\ref{sec:resolution} studies collision order as a concentration and resolution parameter. 
Section~\ref{sec:nonidentifiability} proves that finite collision moments do not identify entropy. 
Section~\ref{sec:experiments} presents numerical experiments, and Section~\ref{sec:discussion} discusses implications for entropy estimation and finite-resolution information representations.

\section{Preliminaries}
\label{sec:preliminaries}

Let $P=(p_1,\ldots,p_m)$ be a discrete probability distribution on a finite alphabet $\mathcal X=\{1,\ldots,m\}$. Throughout, we use logarithms in base $e$, so all entropies are measured in nats. The corresponding expressions in bits are obtained by dividing by $\log 2$.

\subsection{Collision Probabilities}

For $\alpha>0$, define the order-$\alpha$ collision probability, or power sum, of $P$ by
\[
C_\alpha(P)
=
\sum_{i=1}^m p_i^\alpha.
\]
When $\alpha=k$ is a positive integer, $C_k(P)$ has a direct sampling interpretation. 
If $X_1,\ldots,X_k$ are independent samples from $P$, then
\[
C_k(P)
=
\Pr\{X_1=X_2=\cdots=X_k\}.
\]
Thus $C_2(P)$ is the probability of a pairwise collision, $C_3(P)$ is the probability of a triplet collision, and so on.

For a sample $X_1,\ldots,X_N$ from $P$, let $n_i$ denote the number of observations in state $i$. 
For fixed integer $k\geq 2$, the without-replacement empirical collision statistic is
\[
\widehat C_k(P)
=
\frac{\sum_i \binom{n_i}{k}}{\binom{N}{k}}.
\]
Equivalently,
\[
\widehat C_k(P)
=
\binom{N}{k}^{-1}
\sum_{1\leq r_1<\cdots<r_k\leq N}
\mathbf 1\{X_{r_1}=\cdots=X_{r_k}\}.
\]
This is a $U$-statistic \cite{hoeffding1948class,serfling1980approximation}. 
Moreover, it is unbiased for $C_k(P)$, since each unordered $k$-tuple has probability $C_k(P)$ of taking a common value:
\[
\mathbb E[\widehat C_k(P)]
=
\binom{N}{k}^{-1}
\sum_{1\leq r_1<\cdots<r_k\leq N}
\Pr\{X_{r_1}=\cdots=X_{r_k}\}
=
C_k(P).
\]
The fixed-order nature of this estimator will be important below: as $N$ increases, $\widehat C_k(P)$ estimates the same population quantity $C_k(P)$ more accurately, but the target itself remains fixed by the order $k$.

\subsection{R\'enyi Entropy and the Shannon Limit}

For $\alpha>0$, $\alpha\neq 1$, the R\'enyi entropy \cite{renyi1961measures} of order $\alpha$ is
\[
H_\alpha(P)
=
\frac{1}{1-\alpha}\log C_\alpha(P)
=
\frac{1}{1-\alpha}
\log \sum_{i=1}^m p_i^\alpha.
\]
The Shannon entropy is recovered by continuity as $\alpha\to1$ \cite{cover2006elements}:
\[
H(P)
=
\lim_{\alpha\to1}H_\alpha(P)
=
-\sum_{i=1}^m p_i\log p_i.
\]
Indeed, since $C_1(P)=1$, applying l'H\^opital's rule gives
\[
\lim_{\alpha\to1}
\frac{\log C_\alpha(P)}{1-\alpha}
=
-\left.
\frac{d}{d\alpha}\log C_\alpha(P)
\right|_{\alpha=1}
=
-\sum_i p_i\log p_i.
\]

Several special cases of the R\'enyi family are useful for interpretation. 
The order-two entropy
\[
H_2(P)
=
-\log \sum_i p_i^2
\]
is often called collision entropy because it is the negative logarithm of the probability that two independent draws coincide \cite{ribeiro2021entropyuniverse}.
As $\alpha\to\infty$, $H_\alpha(P)$ converges to the min-entropy,
\[
H_\infty(P)
=
-\log \|P\|_\infty
=
-\log \max_i p_i.
\]
Thus the R\'enyi parameter controls how strongly the entropy emphasizes high-probability states.

\subsection{Joint Distributions and Mutual Information}

Let $X$ and $Y$ be discrete random variables with joint distribution
\[
P_{XY}=(p_{xy})_{x\in\mathcal X,\,y\in\mathcal Y}.
\]
Let $P_X=(p_x)$ and $P_Y=(p_y)$ denote the marginals. The Shannon mutual information is
\[
I(X;Y)
=
H(X)+H(Y)-H(X,Y) = \sum_{x,y}p_{xy}
\log
\frac{p_{xy}}{p_xp_y}.
\]

For $\alpha>0$, $\alpha\neq1$, define the marginal and joint collision probabilities
\[
C_\alpha(X)
=
\sum_x p_x^\alpha,
\qquad
C_\alpha(Y)
=
\sum_y p_y^\alpha,
\qquad
C_\alpha(X,Y)
=
\sum_{x,y}p_{xy}^\alpha.
\]
These define the corresponding R\'enyi entropies $H_\alpha(X)$, $H_\alpha(Y)$, and $H_\alpha(X,Y)$.

We define the R\'enyi entropy contrast
\[
K_\alpha(X;Y)
=
H_\alpha(X)+H_\alpha(Y)-H_\alpha(X,Y)=
\frac{1}{1-\alpha}
\left[
\log C_\alpha(X)
+
\log C_\alpha(Y)
-
\log C_\alpha(X,Y)
\right].
\]
This contrast is symmetric in $X$ and $Y$ and invariant under relabeling of the states. 
It also has the Shannon mutual information as its order-one limit:
\[
\lim_{\alpha\to1}K_\alpha(X;Y)
=
I(X;Y).
\]
We refer to $K_\alpha$ as an entropy contrast rather than as a R\'enyi mutual information. 
Several inequivalent R\'enyi generalizations of mutual information exist, with different axiomatic and operational properties \cite{sibson1969information, arimoto1977information, vanerven2014renyi}. 
The contrast used here is the direct analogue of the Shannon identity obtained by replacing Shannon entropies with R\'enyi entropies of common order $\alpha$.
We use $K_\alpha$ only as an entropy-contrast path whose Shannon limit is MI; we do not require it to satisfy axioms proposed for R\'enyi mutual informations.

\section{Low-Order Collision Approximations to Entropy and Mutual Information}
\label{sec:low-order-approximations}

We now use the R\'enyi entropy path to construct finite-order approximations to Shannon entropy. 
The key idea is to treat the integer-order values $H_2(P),H_3(P),H_4(P),\ldots$ as observations of the function $h(\alpha)=H_\alpha(P)$ away from the Shannon point $\alpha=1$. 
Low-order approximations to $H(P)=h(1)$ and their errors are then obtained by extrapolating this path back to $\alpha=1$.

\subsection{Approximating Shannon Entropy}

Let $P=(p_1,\ldots,p_m)$ be a discrete distribution and define $h(\alpha)=H_\alpha(P).$
Since $h(1)=H(P)$, we approximate Shannon entropy by interpolating $h$ at integer orders $\alpha=2,\ldots,r+1$ and evaluating the interpolating polynomial at $\alpha=1$.

For an integer $r\geq 1$, let $L_{r-1}(\alpha)$ be the unique polynomial of degree at most $r-1$ satisfying
\[
L_{r-1}(j+1)=H_{j+1}(P),
\qquad
j=1,\ldots,r.
\]
We define the order-$r$ collision approximation to Shannon entropy by $\widetilde H_r(P) = L_{r-1}(1)$.

Using the Lagrange form of the interpolating polynomial, this approximation can be written explicitly as
\[
\widetilde H_r(P)
=
\sum_{j=1}^{r}
(-1)^{j-1}
\binom{r}{j}
H_{j+1}(P) =
-\sum_{j=1}^{r}
(-1)^{j-1}
\binom{r}{j}
\frac{1}{j}
\log C_{j+1}(P).
\]
Thus the first few approximations are
\[
\widetilde H_1(P)=H_2(P),
\]
\[
\widetilde H_2(P)=2H_2(P)-H_3(P) = -2\log C_2(P)+\frac12\log C_3(P),
\]
and
\[
\widetilde H_3(P)=3H_2(P)-3H_3(P)+H_4(P) = -3\log C_2(P) + \frac32\log C_3(P) - \frac13\log C_4(P).
\]
The subscript $r$ indicates the number of integer R\'enyi entropies used in the approximation. 
Equivalently, $\widetilde H_r$ uses collision probabilities up to order $r+1$.

\subsection{Approximation Error}

The approximation error is deterministic and depends on the regularity of the R\'enyi entropy path near $\alpha=1$. 
The following result is a direct consequence of the interpolation remainder.

\begin{Theorem}[Interpolation error from low-order collision entropies]
\label{thm:entropy-interpolation}
Let $P$ be a discrete distribution and suppose that $h(\alpha)=H_\alpha(P)$ is $r$ times continuously differentiable on $[1,r+1]$. 
Then there exists $\xi\in(1,r+1)$ such that
\[
H(P)-\widetilde H_r(P)
=
(-1)^r h^{(r)}(\xi).
\]
Consequently,
\[
\left|H(P)-\widetilde H_r(P)\right|
\leq
\sup_{\alpha\in[1,r+1]}
\left|h^{(r)}(\alpha)\right|.
\]
\end{Theorem}

\begin{proof}
Let $L_{r-1}$ be the degree-$(r-1)$ interpolating polynomial for $h$ at the points $2,3,\ldots,r+1$. The interpolation remainder gives
\[
h(1)-L_{r-1}(1)
=
\frac{h^{(r)}(\xi)}{r!}
\prod_{j=1}^{r}\left(1-(j+1)\right)
\]
for some $\xi\in(1,r+1)$. Since
\[
\prod_{j=1}^{r}\left(1-(j+1)\right)
=
\prod_{j=1}^{r}(-j)
=
(-1)^r r!,
\]
we obtain
\[
h(1)-L_{r-1}(1)
=
(-1)^r h^{(r)}(\xi).
\]
Using $h(1)=H(P)$ and $L_{r-1}(1)=\widetilde H_r(P)$ proves the stated identity. Taking absolute values gives the bound.
\end{proof}

For the pair--triplet approximation, $\widetilde H_2(P)=2H_2(P)-H_3(P)$, Theorem~\ref{thm:entropy-interpolation} gives
\[
\left|H(P)-\widetilde H_2(P)\right|
\leq
\sup_{\alpha\in[1,3]}
\left|h''(\alpha)\right|.
\]
Likewise, for the pair--triplet--quartet approximation, $\widetilde H_3(P)=3H_2(P)-3H_3(P)+H_4(P)$, we obtain
\[
\left|H(P)-\widetilde H_3(P)\right|
\leq
\sup_{\alpha\in[1,4]}
\left|h^{(3)}(\alpha)\right|.
\]

Theorem~\ref{thm:entropy-interpolation} gives a geometric interpretation of the approximation. 
The quantities $H_2(P),H_3(P), \ldots$ $H_{r+1}(P)$ are values of the R\'enyi entropy path away from the Shannon point $\alpha=1$. 
The approximation $\widetilde H_r(P)$ extrapolates this path back to $\alpha=1$ using a polynomial fitted to those integer-order values. 
The error is therefore controlled by how much the R\'enyi entropy path bends over the interval $[1,r+1]$. 
When the path is nearly polynomial of degree at most $r-1$ on this interval, the approximation is accurate; when the path has substantial higher-order variation, the extrapolation can retain non-negligible deterministic error.

The bound in Theorem~\ref{thm:entropy-interpolation} should be read as a structural error statement rather than a sharp distribution-free guarantee. 
It identifies the source of approximation error: the curvature and higher-order variation of the R\'enyi entropy path between the Shannon point and the integer collision orders used for extrapolation. 
By itself, however, the bound does not imply that a fixed low-order approximation is uniformly accurate across all distributions. 
Sharper error guarantees would require additional assumptions on the distribution, such as bounded support with controlled minimum mass, a specified tail shape, or regularity conditions on the R\'enyi entropy path. 
In the absence of such assumptions, the approximation should be interpreted as a finite-resolution summary whose accuracy depends on the population distribution.

Finally, the error in Theorem~\ref{thm:entropy-interpolation} is not a finite-sample effect.
Rather, it is a property of the population distribution and of the finite set of collision orders used in the approximation. 
Increasing the sample size can reduce estimation error for the collision probabilities, but it cannot remove this interpolation error unless additional collision orders are incorporated or the approximation scheme is changed.

\subsection{Extension to Mutual Information}

The extension to mutual information is direct. 
Let $X$ and $Y$ be discrete random variables with joint distribution $P_{XY}$ and marginals $P_X$ and $P_Y$. 
Define
\[
\widetilde I_r(X;Y)
=
\widetilde H_r(X)
+
\widetilde H_r(Y)
-
\widetilde H_r(X,Y).
\]
Using the definition of $\widetilde H_r$, this can be written as
\[
\widetilde I_r(X;Y)
=
\sum_{j=1}^{r}
(-1)^{j-1}
\binom{r}{j}
\left[
H_{j+1}(X)+H_{j+1}(Y)-H_{j+1}(X,Y)
\right].
\]
Therefore,
\[
\widetilde I_r(X;Y)
=
\sum_{j=1}^{r}
(-1)^{j-1}
\binom{r}{j}
K_{j+1}(X;Y),
\]
where
\[
K_\alpha(X;Y)
=
H_\alpha(X)+H_\alpha(Y)-H_\alpha(X,Y)
\]
is the R\'enyi entropy contrast introduced in Section~\ref{sec:preliminaries}.

The first few mutual-information approximations are
\[
\widetilde I_1(X;Y)=K_2(X;Y),
\]
\[
\widetilde I_2(X;Y)=2K_2(X;Y)-K_3(X;Y),
\]
and
\[
\widetilde I_3(X;Y)=3K_2(X;Y)-3K_3(X;Y)+K_4(X;Y).
\]
These use only low-order marginal and joint collision probabilities.

The interpolation theorem for mutual information follows immediately from the entropy result, or equivalently by applying the same interpolation argument to the R\'enyi contrast path.

\begin{Corollary}[Interpolation error for low-order collision mutual information]
\label{cor:mi-interpolation}
Let $g(\alpha)=K_\alpha(X;Y)$ and suppose that $g$ is $r$ times continuously differentiable on $[1,r+1]$. Then there exists $\xi\in(1,r+1)$ such that
\[
I(X;Y)-\widetilde I_r(X;Y)
=
(-1)^r g^{(r)}(\xi).
\]
Consequently,
\[
\left|I(X;Y)-\widetilde I_r(X;Y)\right|
\leq
\sup_{\alpha\in[1,r+1]}
\left|g^{(r)}(\alpha)\right|.
\]
\end{Corollary}

\begin{proof}
Apply Theorem~\ref{thm:entropy-interpolation} to the contrast path
\[
g(\alpha)=K_\alpha(X;Y).
\]
Since $g(1)=I(X;Y)$ and since $\widetilde I_r(X;Y)$ is the value at $\alpha=1$ of the degree-$(r-1)$ polynomial interpolating $g$ at $\alpha=2,\ldots,r+1$, the same interpolation remainder gives the result.
\end{proof}

Thus the mutual-information approximation is not a separate construction. 
It is the entropy approximation applied to the marginal and joint distributions. 
The population quantity $\widetilde I_r(X;Y)$ depends only on the finite collection of collision probabilities
\[
\{C_k(X),C_k(Y),C_k(X,Y): k=2,\ldots,r+1\}.
\]

In collision-probability form,
\[
K_k(X;Y)
=
\frac{1}{1-k}
\left[
\log C_k(X)+\log C_k(Y)-\log C_k(X,Y)
\right],
\]
so
\[
\widetilde I_r(X;Y)
=
-\sum_{j=1}^{r}
(-1)^{j-1}
\binom{r}{j}
\frac{1}{j}
\left[
\log C_{j+1}(X)
+
\log C_{j+1}(Y)
-
\log C_{j+1}(X,Y)
\right].
\]
This expression makes explicit that $\widetilde I_r(X;Y)$ uses low-order coincidence information from the two marginals and the joint distribution.

\subsection{Independence Baseline}

The collision-based mutual-information approximations inherit the independence baseline of mutual information at the population level. 
If $X$ and $Y$ are independent, then
\[
p_{xy}=p_xp_y.
\]
Therefore, for every $\alpha>0$,
\[
C_\alpha(X,Y)
=
\sum_{x,y}(p_xp_y)^\alpha
=
\left(\sum_x p_x^\alpha\right)
\left(\sum_y p_y^\alpha\right)
=
C_\alpha(X)C_\alpha(Y).
\]
It follows that
\[
K_\alpha(X;Y)=0
\]
for every $\alpha\neq1$, and hence
\[
\widetilde I_r(X;Y)=0
\]
for every approximation order $r$.

This property should be interpreted carefully. 
The contrast $K_\alpha$ is not intended as a full R\'enyi mutual information in the axiomatic sense; several inequivalent R\'enyi generalizations of mutual information exist, with different operational and mathematical properties. 
Here, $K_\alpha$ is used because it is the direct analogue of the Shannon identity $I(X;Y)=H(X)+H(Y)-H(X,Y)$ obtained by replacing Shannon entropies with R\'enyi entropies of common order $\alpha$. 
Its role is to connect low-order collision statistics to the Shannon mutual information limit.

For fixed $r$, $\widetilde I_r(X;Y)$ is therefore a finite-resolution dependence functional. 
It agrees with the independence baseline, uses only low-order collision information, and approximates Shannon mutual information when the R\'enyi contrast path is sufficiently regular near $\alpha=1$.

\subsection{Finite Moment Truncation}

The finite-order nature of $\widetilde H_r$ and $\widetilde I_r$ is both a strength and a limitation.
It is a strength because the required ingredients are low-order collision probabilities, which are directly estimable from coincidence counts. 
It is a limitation because finitely many collision probabilities do not generally determine Shannon entropy.

For a distribution on $m$ symbols, sufficiently many power sums determine the distribution up to permutation. 
However, a finite truncation $C_2(P),\ldots,C_{r+1}(P)$ does not generally identify $P$ when the alphabet is large enough. 
Consequently, it cannot generally determine $H(P)$. 
We return to this point in Section~\ref{sec:nonidentifiability}, where we prove a non-identifiability result for finite collision-moment truncations.
The purpose of the low-order approximation is therefore not to eliminate the need for Shannon entropy estimators. 
Rather, it provides a controlled way to ask what part of Shannon entropy and mutual information is visible from a finite number of collision statistics, how that information changes as additional collision orders are included, and how estimation error differs from approximation error.

\section{Fixed-Order Asymptotics}
\label{sec:fixed-order-asymptotics}

The preceding section defined population approximations $\widetilde H_r(P)$ and $\widetilde I_r(X;Y)$ from finitely many collision orders. 
We now consider their empirical counterparts. 
The main point is that, for fixed approximation order, increasing the sample size improves estimation of the finite-resolution target, but it does not remove the deterministic approximation error relative to Shannon entropy or mutual information.

\subsection{Empirical Collision Entropies}

Let $X_1,\ldots,X_N$ be i.i.d. samples from a discrete distribution $P$. 
For fixed integer $k\geq2$, let
\[
\widehat C_k(P)
=
\frac{\sum_i \binom{n_i}{k}}{\binom{N}{k}}
\]
be the empirical collision probability. The corresponding empirical collision entropy is
\[
\widehat H_k(P)
=
\frac{1}{1-k}\log \widehat C_k(P),
\]
whenever $\widehat C_k(P)>0$.

For a fixed approximation order $r$, define the empirical low-order entropy approximation by
\[
\widehat{\widetilde H}_r(P)
=
\sum_{j=1}^{r}
(-1)^{j-1}
\binom{r}{j}
\widehat H_{j+1}(P).
\]
Equivalently,
\[
\widehat{\widetilde H}_r(P)
=
-\sum_{j=1}^{r}
(-1)^{j-1}
\binom{r}{j}
\frac{1}{j}
\log \widehat C_{j+1}(P).
\]
This estimator uses empirical collision probabilities up to order $r+1$.

Since $\widehat C_k(P)$ is an unbiased and consistent estimator of $C_k(P)$ for fixed $k$, the continuous mapping theorem gives
\[
\widehat H_k(P)\xrightarrow{p}H_k(P)
\]
provided $C_k(P)>0$. 
Therefore,
\[
\widehat{\widetilde H}_r(P)
\xrightarrow{p}
\widetilde H_r(P)
\]
for fixed $r$ as $N\to\infty$.

Thus $\widehat{\widetilde H}_r(P)$ is consistent for the finite-order approximation $\widetilde H_r(P)$, not necessarily for the Shannon entropy $H(P)$. 
This gives the decomposition
\[
\widehat{\widetilde H}_r(P)-H(P)
=
\underbrace{
\widehat{\widetilde H}_r(P)-\widetilde H_r(P)
}_{\text{estimation error}}
+
\underbrace{
\widetilde H_r(P)-H(P)
}_{\text{approximation error}}.
\]
The first term vanishes in probability as $N\to\infty$ for fixed $r$. 
The second term is the interpolation error studied in Theorem~\ref{thm:entropy-interpolation} and remains fixed for a given population distribution and approximation order.

\subsection{Asymptotic Behavior for Fixed Order}

The fixed-order estimation error has standard $U$-statistic behavior. 
For fixed $k$, the empirical collision probability $\widehat C_k(P)$ is a bounded $U$-statistic. 
Under the usual nondegeneracy conditions,
\[
\sqrt{N}\left(\widehat C_k(P)-C_k(P)\right)
\Rightarrow
\mathcal N(0,\sigma_k^2),
\]
where $\sigma_k^2$ is the asymptotic variance determined by the first projection of the $U$-statistic kernel.

Applying the delta method to the map
\[
c\mapsto \frac{1}{1-k}\log c
\]
gives
\[
\sqrt{N}\left(\widehat H_k(P)-H_k(P)\right)
\Rightarrow
\mathcal N(0,\eta_k^2),
\]
with
\[
\eta_k^2
=
\frac{\sigma_k^2}{(1-k)^2 C_k(P)^2}.
\]
Because $\widehat{\widetilde H}_r(P)$ is a finite linear combination of the empirical entropies $\widehat H_2(P),\ldots,\widehat H_{r+1}(P)$, it follows that
\[
\sqrt{N}
\left(
\widehat{\widetilde H}_r(P)-\widetilde H_r(P)
\right)
\Rightarrow
\mathcal N(0,\eta_r^2),
\]
for some finite asymptotic variance $\eta_r^2$, again under the corresponding nondegeneracy conditions.

Combining this with the deterministic interpolation error gives
\[
\widehat{\widetilde H}_r(P)-H(P)
=
O_p(N^{-1/2})
+
\bigl(\widetilde H_r(P)-H(P)\bigr).
\]
This expression separates the two sources of error. 
The stochastic term decreases with sample size. 
The deterministic term depends on the finite collision orders used and is not reduced by additional samples.

\subsection{Extension to Mutual Information}

The same decomposition applies to mutual information. 
Let $(X_1,Y_1),\ldots,(X_N,Y_N)$ be i.i.d. samples from a joint distribution $P_{XY}$. 
Define empirical collision probabilities for the marginals and the joint distribution:
\[
\widehat C_k(X)
=
\frac{\sum_x \binom{n_x}{k}}{\binom{N}{k}},
\qquad
\widehat C_k(Y)
=
\frac{\sum_y \binom{n_y}{k}}{\binom{N}{k}},
\]
and
\[
\widehat C_k(X,Y)
=
\frac{\sum_{x,y} \binom{n_{xy}}{k}}{\binom{N}{k}}.
\]
The empirical R\'enyi entropy contrast is
\[
\widehat K_k(X;Y)
=
\widehat H_k(X)+\widehat H_k(Y)-\widehat H_k(X,Y),
\]
or equivalently
\[
\widehat K_k(X;Y)
=
\frac{1}{1-k}
\left[
\log \widehat C_k(X)
+
\log \widehat C_k(Y)
-
\log \widehat C_k(X,Y)
\right].
\]

For fixed approximation order $r$, define
\[
\widehat{\widetilde I}_r(X;Y)
=
\sum_{j=1}^{r}
(-1)^{j-1}
\binom{r}{j}
\widehat K_{j+1}(X;Y).
\]
Then, for fixed $r$,
\[
\widehat{\widetilde I}_r(X;Y)
\xrightarrow{p}
\widetilde I_r(X;Y)
\]
as $N\to\infty$, provided the relevant population collision probabilities are nonzero.

Hence the total error decomposes as
\[
\widehat{\widetilde I}_r(X;Y)-I(X;Y)
=
\underbrace{
\widehat{\widetilde I}_r(X;Y)-\widetilde I_r(X;Y)
}_{\text{estimation error}}
+
\underbrace{
\widetilde I_r(X;Y)-I(X;Y)
}_{\text{approximation error}}.
\]
For fixed $r$, the first term is stochastic and decreases at the usual $N^{-1/2}$ rate under standard nondegeneracy conditions. 
The second term is deterministic and is controlled by the interpolation error in Corollary~\ref{cor:mi-interpolation}.

Thus,
\[
\widehat{\widetilde I}_r(X;Y)-I(X;Y)
=
O_p(N^{-1/2})
+
\bigl(\widetilde I_r(X;Y)-I(X;Y)\bigr).
\]

\subsection{Interpretation of asymptotic behavior}

The fixed-order asymptotic behavior clarifies the role of sample size. Increasing $N$ makes empirical collision probabilities more accurate estimates of their population counterparts. It does not make a fixed collision order more Shannon-like. For example,
\[
\widehat{\widetilde I}_2(X;Y)
=
2\widehat K_2(X;Y)-\widehat K_3(X;Y)
\]
converges to
\[
\widetilde I_2(X;Y)
=
2K_2(X;Y)-K_3(X;Y),
\]
not necessarily to $I(X;Y)$.

Consequently, a fixed-order collision approximation can exhibit a persistent bias relative to Shannon entropy or mutual information even when the sample size is large. 
This persistent bias is not a failure of estimation. 
It is the deterministic approximation error associated with using finitely many collision orders. 
To reduce this component, one must change the approximation itself, for example by incorporating additional collision orders or by using a different extrapolation scheme.

This distinction is central to the finite-resolution interpretation. 
The sample size controls statistical precision, while the collision orders determine the population target. 
In empirical comparisons, collision-based approximations should therefore be evaluated in two ways: relative to the Shannon quantity, which measures total approximation-plus-estimation error, and relative to their finite-order population targets, which isolates statistical estimation error.

\section{Collision Order as a Concentration and Resolution Parameter}
\label{sec:resolution}

The collision order does more than index a family of estimators. It also changes which parts of the distribution are emphasized. 
Low orders aggregate probability mass broadly across the support, while high orders increasingly concentrate on the largest probabilities. 
This section makes that statement precise and interprets the order $k$ as a resolution parameter.

Let $P=(p_1,\ldots,p_m)$ be a discrete distribution, with probabilities ordered as
\[
p_1\geq p_2\geq \cdots \geq p_m\geq 0.
\]
Recall that the order-$k$ collision probability is
\[
C_k(P)=\sum_{i=1}^m p_i^k.
\]
For larger $k$, the terms $p_i^k$ magnify differences among probabilities: high-probability states remain visible, while lower-probability states are suppressed.

\subsection{Convergence Toward the Largest Probability}

The following elementary result shows that the collision probability becomes dominated by the largest mass as the order increases.

\begin{Lemma}[Collision roots converge to the maximum mass]
\label{lem:ck-root}
Let $P=(p_1,\ldots,p_m)$ be a discrete distribution and let
\[
p_{\max}=\max_i p_i.
\]
Then
\[
\lim_{k\to\infty} C_k(P)^{1/k}=p_{\max}.
\]
\end{Lemma}

\begin{proof}
Since $p_{\max}^k$ is one term in the sum defining $C_k(P)$,
\[
p_{\max}^k
\leq
C_k(P).
\]
Also, since each $p_i\leq p_{\max}$,
\[
C_k(P)
=
\sum_i p_i^k
\leq
m p_{\max}^k.
\]
Taking $k$-th roots gives
\[
p_{\max}
\leq
C_k(P)^{1/k}
\leq
m^{1/k}p_{\max}.
\]
Since $m^{1/k}\to1$ as $k\to\infty$, the result follows.
\end{proof}

This result is the collision-probability form of the standard convergence of R\'enyi entropy to min-entropy:
\[
H_k(P)
=
\frac{1}{1-k}\log C_k(P)
\longrightarrow
-\log p_{\max}
=
H_\infty(P).
\]
Thus, increasing $k$ moves the entropy functional away from the Shannon regime and toward a quantity determined by the largest atom of the distribution.

\subsection{Tail Suppression}

The preceding lemma describes the limiting behavior as $k\to\infty$. The next bound quantifies how quickly lower-probability states lose influence.

For $s<m$, define the top-$s$ contribution to the order-$k$ collision probability as
\[
R_s(k)
=
\frac{\sum_{i=1}^{s}p_i^k}{\sum_{i=1}^{m}p_i^k}.
\]
Thus $R_s(k)$ is the fraction of the collision probability contributed by the $s$ largest states.

\begin{Lemma}[Exponential suppression of lower-probability states]
\label{lem:tail-suppression}
Let $P=(p_1,\ldots,p_m)$ be ordered so that
\[
p_1\geq p_2\geq \cdots \geq p_m.
\]
For any $s<m$ with $p_s>0$,
\[
1-R_s(k)
\leq
\frac{m-s}{s}
\left(\frac{p_{s+1}}{p_s}\right)^k.
\]
In particular, if $p_s>p_{s+1}$, the contribution of the tail beyond the top $s$ states decays exponentially in $k$.
\end{Lemma}

\begin{proof}
We have
\[
1-R_s(k)
=
\frac{\sum_{i=s+1}^{m}p_i^k}{\sum_{i=1}^{m}p_i^k}.
\]
Since the denominator is at least the contribution of the top $s$ states,
\[
1-R_s(k)
\leq
\frac{\sum_{i=s+1}^{m}p_i^k}{\sum_{i=1}^{s}p_i^k}.
\]
For $i>s$, $p_i\leq p_{s+1}$, so
\[
\sum_{i=s+1}^{m}p_i^k
\leq
(m-s)p_{s+1}^k.
\]
For $i\leq s$, $p_i\geq p_s$, so
\[
\sum_{i=1}^{s}p_i^k
\geq
s p_s^k.
\]
Combining the two inequalities gives
\[
1-R_s(k)
\leq
\frac{m-s}{s}
\left(\frac{p_{s+1}}{p_s}\right)^k.
\]
If $p_s>p_{s+1}$, then $p_{s+1}/p_s<1$, so the bound decays exponentially in $k$.
\end{proof}

Lemma~\ref{lem:tail-suppression} provides a finite-order version of the concentration interpretation. 
The order $k$ determines how sharply the collision probability distinguishes high-mass states from lower-mass states. 
When the leading probabilities are separated from the tail, increasing $k$ rapidly concentrates the collision statistic on those leading states.

\subsection{Effective Support at Order \texorpdfstring{$k$}{k}}

The concentration behavior can also be expressed through the effective support associated with R\'enyi entropy. 
Define
\[
S_k(P)
=
\exp(H_k(P)).
\]
For $k\neq1$,
\[
S_k(P)
=
C_k(P)^{1/(1-k)}.
\]
This quantity is often interpreted as an effective number of states at order $k$. 
At $k=1$, it reduces by continuity to the Shannon effective support,
\[
S_1(P)=\exp(H(P)).
\]
At $k=2$,
\[
S_2(P)=\frac{1}{\sum_i p_i^2},
\]
the reciprocal of the pair-collision probability. As $k\to\infty$,
\[
S_k(P)\to \frac{1}{p_{\max}}.
\]

Thus increasing $k$ decreases the effective support from a Shannon-weighted notion of diversity toward a support size determined by the largest probability. 
This gives another interpretation of collision order as a resolution parameter: larger values of $k$ resolve the distribution at a scale dominated by its largest atoms.

\subsection{Implications for Low-Order Approximations}

The low-order approximations introduced above use the values $H_2(P),H_3(P),\ldots,H_{r+1}(P)$.
The concentration results show that these values are not interchangeable observations of the same feature of the distribution. 
Each order weights the distribution differently. 
The order-two entropy is sensitive to pair collisions and reflects a relatively broad notion of concentration. 
Higher orders increasingly emphasize repeated observations from the largest probability masses.

This has two consequences. 
First, adding higher collision orders can improve the extrapolation to Shannon entropy if those orders provide useful information about the shape of the R\'enyi entropy path near $\alpha=1$. 
Second, adding higher orders also shifts the approximation toward statistics that are more sensitive to high-mass events and less sensitive to the tail. 
In finite samples, higher-order collisions are also harder to estimate, because they require repeated observations of the same state among larger tuples.

Thus the approximation order controls a tradeoff. 
Low-order approximations use statistics that are relatively stable and broadly distributed across the support, but they may retain larger interpolation error.
Higher-order approximations can better capture curvature in the R\'enyi entropy path, but they rely on collision probabilities that are more concentrated and may have higher sampling variability.

\section{Non-Identifiability from Finite Collision Moments}
\label{sec:nonidentifiability}

The preceding sections show how finite collections of collision probabilities can be used to approximate Shannon entropy. 
We now show that such finite collections cannot, in general, determine Shannon entropy exactly. 
This provides a formal limitation on any method that uses only finitely many low-order collision moments.

The key point is that collision probabilities are power sums of the probability vector. 
A sufficiently long sequence of power sums determines a finite distribution up to permutation, but a finite truncation generally leaves degrees of freedom. 
Along those remaining directions, Shannon entropy can vary.

\subsection{An Illustrative Example}

The simplest example shows that the pair-collision probability alone does not determine entropy. 
Consider the two distributions $P=\left(\frac12,\frac12,0\right)$ and $Q=\left(\frac23,\frac16,\frac16\right)$.
Both are distributions on three symbols. 
Their pair-collision probabilities are equal:
\[
C_2(P)
=
\left(\frac12\right)^2+\left(\frac12\right)^2+0^2
=
\frac12,
\]
while
\[
C_2(Q)
=
\left(\frac23\right)^2
+
\left(\frac16\right)^2
+
\left(\frac16\right)^2
=
\frac49+\frac1{36}+\frac1{36}
=
\frac12.
\]
Thus the two distributions are indistinguishable by their order-two collision probability.

However, their Shannon entropies differ: $H(P)=\log 2$, whereas $H(Q) = -\frac23\log\frac23 - 2\cdot\frac16\log\frac16$.
Therefore, knowledge of $C_2(P)$ alone is insufficient to determine $H(P)$.

The same example also shows how higher-order collisions can reveal distinctions hidden from pair collisions. The order-three collision probabilities are
\[
C_3(P)
=
\left(\frac12\right)^3+\left(\frac12\right)^3
=
\frac14,
\]
whereas
\[
C_3(Q)
=
\left(\frac23\right)^3
+
\left(\frac16\right)^3
+
\left(\frac16\right)^3
=
\frac{8}{27}
+
\frac{1}{216}
+
\frac{1}{216}
=
\frac{11}{36}.
\]
Thus $C_2(P)=C_2(Q)$, but $C_3(P)\neq C_3(Q)$. Pair collisions do not determine triplet collisions, and neither alone determines the full entropy.

\subsection{General Non-Identifiability}

The preceding example is not special. 
For any fixed number of collision moments, there are sufficiently large alphabets on which those moments do not identify Shannon entropy.

\begin{Theorem}[Finite collision moments do not determine entropy]
\label{thm:nonidentifiability}
Let $r\geq 2$ and let $m\geq r+1$. Then there exist two distinct probability distributions $P$ and $Q$ on $m$ symbols such that
\[
C_k(P)=C_k(Q),
\qquad
k=2,\ldots,r,
\]
but
\[
H(P)\neq H(Q).
\]
Consequently, Shannon entropy is not determined by the finite collision sequence
\[
C_2(P),C_3(P),\ldots,C_r(P)
\]
on the $m$-symbol simplex when $m\geq r+1$.
\end{Theorem}

\begin{proof}
Let
\[
\Delta_m^\circ
=
\left\{
p\in\mathbb R^m:
p_i>0,\ \sum_{i=1}^m p_i=1
\right\}
\]
denote the interior of the $m$-symbol probability simplex. Define the map
\[
F:\Delta_m^\circ\to\mathbb R^{r-1}
\]
by
\[
F(p)
=
\left(
C_2(p),C_3(p),\ldots,C_r(p)
\right).
\]
Equivalently, include the normalization condition and define
\[
G(p)
=
\left(
C_1(p),C_2(p),\ldots,C_r(p)
\right),
\]
where
\[
C_1(p)=\sum_{i=1}^m p_i=1.
\]

For $\ell=1,\ldots,r$, the gradient of $C_\ell$ in $\mathbb R^m$ is
\[
\nabla C_\ell(p)
=
\ell
\left(
p_1^{\ell-1},
p_2^{\ell-1},
\ldots,
p_m^{\ell-1}
\right).
\]
Choose a point $p\in\Delta_m^\circ$ with at least $r$ distinct coordinates. Then the matrix with rows
\[
\left(
p_1^{\ell-1},
p_2^{\ell-1},
\ldots,
p_m^{\ell-1}
\right),
\qquad
\ell=1,\ldots,r,
\]
has rank $r$ by the Vandermonde rank condition. Hence $G$ has rank $r$ at $p$.

Because $m\geq r+1$, the level set of $G$ through $p$ has positive local dimension. By the implicit function theorem, there is a smooth local manifold
\[
\mathcal M
=
\left\{
q\in\Delta_m^\circ:
C_k(q)=C_k(p),\ k=2,\ldots,r
\right\}
\]
of dimension at least $m-r\geq1$ near $p$.

It remains to show that Shannon entropy is not constant on this level set. The gradient of Shannon entropy is
\[
\nabla H(p)
=
-\left(
\log p_1+1,
\ldots,
\log p_m+1
\right).
\]
If $H$ were locally constant on $\mathcal M$, then $\nabla H(p)$ would lie in the span of the gradients
\[
\nabla C_1(p),\nabla C_2(p),\ldots,\nabla C_r(p).
\]
Equivalently, there would exist constants $a_0,\ldots,a_{r-1}$ such that
\[
-(\log p_i+1)
=
a_0+a_1p_i+\cdots+a_{r-1}p_i^{r-1},
\qquad
i=1,\ldots,m.
\]
That is, the values of $-\log t-1$ at the points $p_i$ would agree with a polynomial of degree at most $r-1$.

Choose $p$ with at least $r+1$ distinct coordinates such that this interpolation relation does not hold. Such a choice is possible because $-\log t-1$ is not a polynomial of degree at most $r-1$. Therefore,
\[
\nabla H(p)
\notin
\operatorname{span}
\{\nabla C_1(p),\ldots,\nabla C_r(p)\}.
\]
Hence there exists a tangent vector $v$ to the level set $\mathcal M$ at $p$ such that
\[
\nabla H(p)\cdot v\neq0.
\]

By the implicit function theorem, there exists a smooth curve $p(t)\subset\mathcal M$ with $p(0)=p$ and $p'(0)=v$. Along this curve,
\[
C_k(p(t))=C_k(p),
\qquad
k=2,\ldots,r,
\]
but
\[
\left.
\frac{d}{dt}H(p(t))
\right|_{t=0}
=
\nabla H(p)\cdot v
\neq0.
\]
Thus, for sufficiently small $t\neq0$, setting $Q=p(t)$ and $P=p$ gives
\[
C_k(P)=C_k(Q),
\qquad
k=2,\ldots,r,
\]
but
\[
H(P)\neq H(Q).
\]
This proves the claim.
\end{proof}

\subsection{Interpretation of Non-identifiability}

Theorem~\ref{thm:nonidentifiability} shows that the approximation error identified earlier is unavoidable in general. 
A finite collection of low-order collision probabilities defines an equivalence class of distributions. 
Within that equivalence class, Shannon entropy may vary. 
Therefore, no estimator or approximation that depends only on $C_2(P),\ldots,C_r(P)$ can recover $H(P)$ exactly over all distributions on sufficiently large alphabets.

This does not mean that low-order collision approximations are uninformative. 
Rather, it clarifies what they can and cannot do. 
They provide finite-resolution summaries of the distribution, and these summaries may approximate Shannon entropy well when the R\'enyi entropy path is sufficiently regular. 
But the finite collision sequence does not contain all information needed to determine Shannon entropy in general.

The theorem also explains why increasing sample size alone cannot remove the residual approximation error. 
Even if the collision probabilities
\[
C_2(P),\ldots,C_r(P)
\]
were known exactly, they would not generally identify $H(P)$. 
Additional samples improve estimation of those moments, but additional collision orders are needed to further refine the finite-resolution representation.

\subsection{Relation to Full Moment Recovery}

The non-identifiability result concerns finite truncations of the power-sum sequence. 
It is compatible with the classical fact that sufficiently many power sums determine a finite distribution up to permutation. 
For a distribution on $m$ symbols, the power sums
\[
C_1(P),C_2(P),\ldots,C_m(P)
\]
determine the elementary symmetric polynomials through Newton identities, and hence determine the multiset
\[
\{p_1,\ldots,p_m\}
\]
up to permutation. 
Once the probability multiset is determined, Shannon entropy is determined.

Thus the full finite-alphabet distribution can, in principle, be recovered from enough collision moments. 
The limitation arises when only a fixed low-order truncation is used. 
The framework developed in this paper sits between these two extremes: it uses a small number of directly estimable collision moments to approximate Shannon quantities, while making explicit the residual uncertainty induced by finite moment truncation.

\section{Numerical Experiments}
\label{sec:experiments}

We now illustrate the behavior of the collision-based approximations in controlled discrete settings. 
The experiments are designed to separate three effects: approximation error from using finitely many collision orders, estimation error from finite samples, and the concentration behavior induced by increasing collision order. 
The goal is not to show that collision-based approximations uniformly outperform direct estimators of Shannon entropy or mutual information. 
Rather, the experiments demonstrate the approximation--estimation tradeoff developed in the previous sections.

Throughout, logarithms are natural and all entropies are reported in nats. 
We compare the collision-based approximations with the empirical plug-in estimator and the Miller--Madow correction which target Shannon entropy or mutual information directly. 
The collision-based approximations instead target finite-resolution quantities constructed from low-order collision statistics.
Note that the subscript $r$ on $\widetilde I_r$ denotes the number of integer R\'enyi contrast values used, the largest collision order is thus $r+1$.

\subsection{Population Approximation Error}
\label{sec:population-approximation-error}

We first examine approximation error at the population level, where the joint distribution is known exactly and there is no sampling variability. 
This isolates the deterministic gap between Shannon mutual information and its finite-order collision approximations.

We begin with the binary symmetric channel. 
Let \(X\sim \mathrm{Bernoulli}(1/2)\) and let \(Y=X\oplus E\), where \(E\sim\mathrm{Bernoulli}(\epsilon)\) is independent noise. 
The joint distribution is
\[
P_{XY}
=
\begin{pmatrix}
(1-\epsilon)/2 & \epsilon/2 \\
\epsilon/2 & (1-\epsilon)/2
\end{pmatrix}.
\]
Varying \(\epsilon\) from \(0\) to \(1/2\) interpolates between perfect dependence and independence. 
For each value of \(\epsilon\), we compute the Shannon mutual information \(I(X;Y)\) and the finite-resolution collision approximations \(\widetilde I_2\), \(\widetilde I_3\), and \(\widetilde I_4\). 
Because these quantities are evaluated directly from the population distribution, any discrepancy between \(I\) and \(\widetilde I_r\) is deterministic approximation error, not sampling error.

Figure~\ref{fig:bsc}A shows that the collision approximations reproduce the qualitative dependence structure of the binary symmetric channel. 
All curves decrease from the deterministic limit at \(\epsilon=0\) to the independent limit at \(\epsilon=1/2\), and the finite-order approximations are exact at both endpoints. 
For intermediate noise levels, however, the approximations overestimate Shannon mutual information. 
This discrepancy is not caused by finite samples; it reflects the fact that a finite set of low-order collision statistics captures only a finite-resolution view of the R\'enyi contrast path.

Figure~\ref{fig:bsc}B shows the corresponding absolute errors. 
The error is small near the endpoints and largest at intermediate noise levels, where the curve relating the integer-order R\'enyi contrasts to the Shannon point has the greatest effective curvature. 
Increasing the approximation order reduces the error in this example: \(\widetilde I_2\) has the largest error, \(\widetilde I_3\) is smaller, and \(\widetilde I_4\) is smaller still. 
Thus, adding collision orders can improve the population approximation, but it does not remove the conceptual distinction between the Shannon target and the finite-resolution target. 
The remaining gap is a property of the finite-order approximation itself.

\begin{figure}[H]
\isPreprints{\centering}{} 
\includegraphics[width=\textwidth]{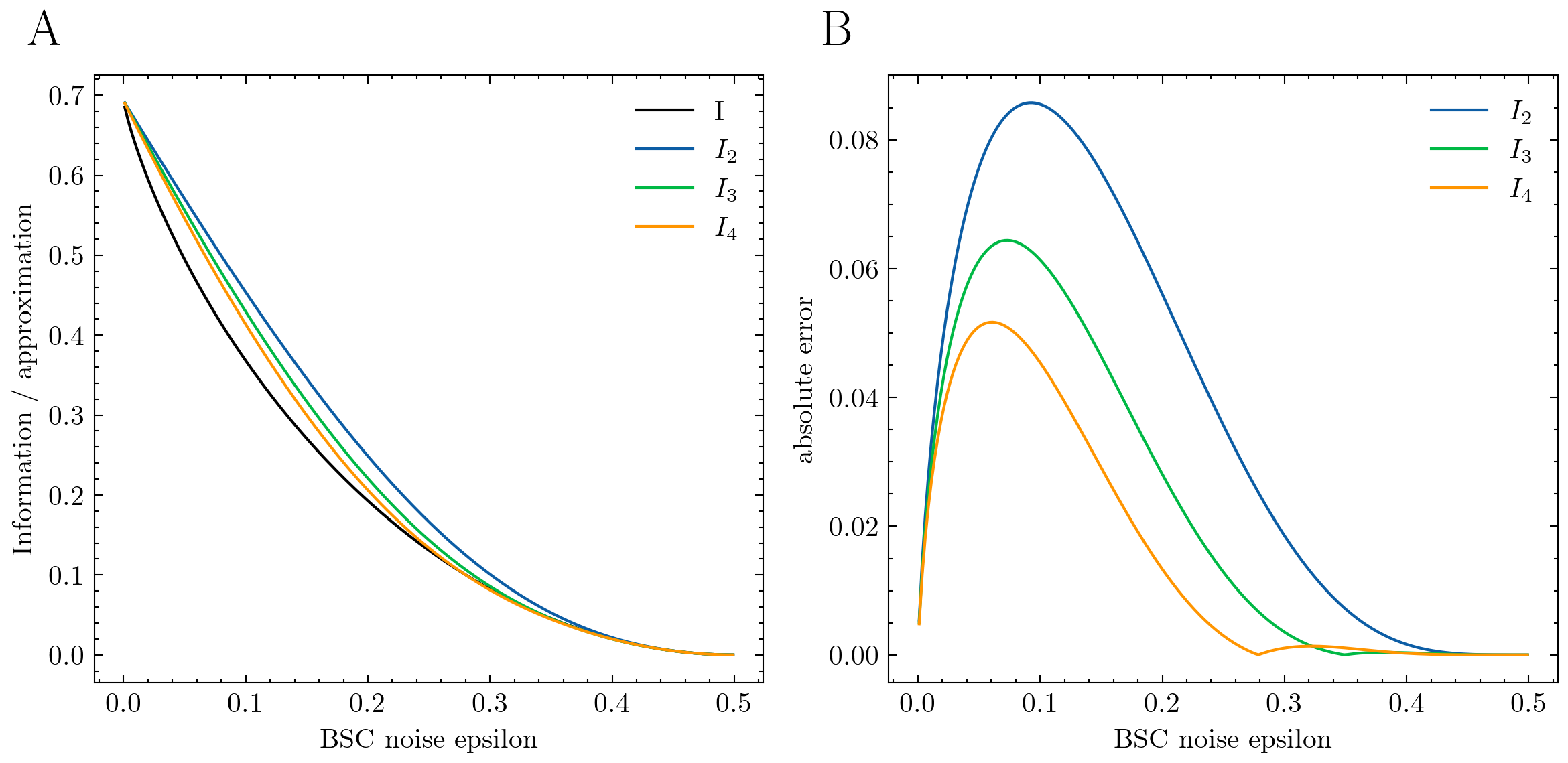}
\caption{
\textbf{Population approximation error for the binary symmetric channel}. 
A) compares the analytic Shannon mutual information \(I(X;Y)\) with finite-resolution collision approximations \(\widetilde I_2\), \(\widetilde I_3\), and \(\widetilde I_4\) as the noise level \(\epsilon\) varies from near perfect dependence to independence. 
All approximations match the endpoint behavior but overestimate \(I(X;Y)\) at intermediate noise levels. 
B) the corresponding absolute approximation errors. 
The error is largest at intermediate noise and decreases with approximation order in this example, illustrating the deterministic interpolation error that remains even when the collision probabilities are known exactly.
}
\label{fig:bsc}
\end{figure}   
\unskip

\subsection{Fixed-Order Estimation and the Approximation--Estimation Split}
\label{sec:fixed-order-sampling-experiment}

The population experiment above isolates deterministic approximation error. 
We next add sampling variability to separate approximation error from estimation error. 
For a fixed collision order \(r\), the empirical estimator \(\widehat{\widetilde I}_r\) is computed from finite samples using the without-replacement collision statistics. 
As sample size increases, these empirical collision statistics converge to their population collision probabilities, and therefore \(\widehat{\widetilde I}_r\) converges to the finite-resolution population target \(\widetilde I_r\). 
Importantly, this target need not equal the Shannon mutual information \(I\). 
Thus increasing \(N\) removes estimation error for the fixed-order collision approximation, but it does not remove the deterministic approximation error between \(\widetilde I_r\) and \(I\).

We illustrate this distinction using the binary symmetric channel with \(\epsilon=0.2\), for which the Shannon mutual information is known analytically. 
For each sample size \(N\), we generate repeated samples from the channel and estimate \(\widehat{\widetilde I}_2,\widehat{\widetilde I}_3,\widehat{\widetilde I}_4,\widehat{\widetilde I}_5\). 
We compare these estimates with their corresponding population targets \(\widetilde I_2,\widetilde I_3,\widetilde I_4,\widetilde I_5\), shown as horizontal dotted lines. 
We also include standard plug-in and Miller--Madow estimators of Shannon mutual information, with the analytic Shannon value shown as a dashed line.

Figure~\ref{fig:bsc-sampling} shows the expected fixed-order behavior. 
For each \(r\), the empirical collision estimator approaches its own population target as \(N\) increases. 
The sampling variability, shown by error bars, decreases with sample size, and the high-order estimators require larger samples before stabilizing. 
However, the limiting targets remain separated from the analytic Shannon mutual information. 
The order-two estimate stabilizes around \(\widetilde I_2\), not around \(I\); the same pattern holds for the higher orders. 
Increasing the collision order moves the population target closer to \(I\) in this example, but does not make the estimator a conventional Shannon estimator. 
This demonstrates the key approximation--estimation distinction: fixed-order collision methods consistently estimate finite-resolution information targets, while Shannon estimators such as the plug-in and Miller--Madow estimators aim directly at \(I\).

The high-order estimates also show the cost of increasing resolution. 
In this example, larger \(r\) moves the finite-resolution target closer to the Shannon mutual information, but the corresponding empirical estimator is less stable at small sample sizes. 
This is expected: higher-order collision statistics depend on rarer coincidence events, so they require more data before their empirical estimates stabilize. 
Thus increasing collision order can reduce deterministic approximation error while increasing finite-sample estimation error, making the approximation--estimation tradeoff visible in the same figure.

\begin{figure}[t]
\centering
\includegraphics[width=\textwidth]{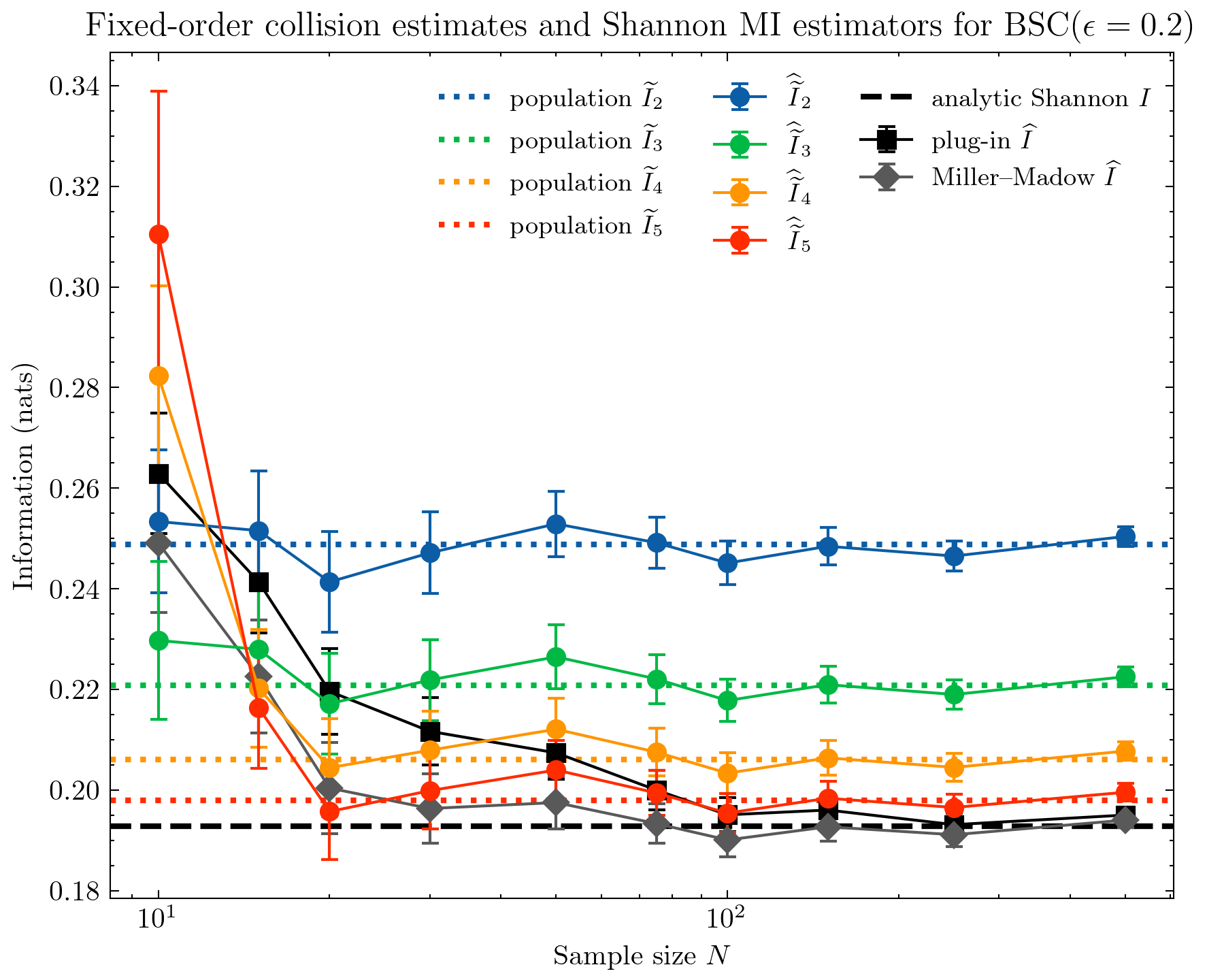}
\caption{
\textbf{Fixed-order sampling behavior for the binary symmetric channel with \(\epsilon=0.2\)}. 
For each sample size \(N\), repeated samples are drawn from the channel and used to estimate finite-resolution collision approximations \(\widehat{\widetilde I}_2,\widehat{\widetilde I}_3,\widehat{\widetilde I}_4,\widehat{\widetilde I}_5\). 
Solid colored curves show the empirical means, with error bars indicating two standard errors. 
Dotted horizontal lines show the corresponding population finite-resolution targets \(\widetilde I_2,\widetilde I_3,\widetilde I_4,\widetilde I_5\), computed from the exact joint distribution. 
The dashed black line shows the analytic Shannon mutual information \(I\), while the black and gray curves show plug-in and Miller--Madow estimates of \(I\). 
As \(N\) increases, each collision estimator converges to its own finite-resolution target rather than to the Shannon value. 
Thus sampling error vanishes with increasing \(N\), but the deterministic approximation gap between \(\widetilde I_r\) and \(I\) remains for fixed \(r\).
}
\label{fig:bsc-sampling}
\end{figure}
\unskip

\subsection{Collision Order as a Concentration Filter}
\label{sec:collision-resolution-experiment}

The collision order \(k\) controls which parts of a distribution are emphasized. 
To illustrate this effect, we constructed categorical distributions with the same support size but different concentration profiles: a uniform distribution and three increasingly tiered distributions in which probability mass is concentrated into progressively smaller high-probability groups. 
For each distribution, we computed the normalized contribution of each state to the order-\(k\) collision probability,
\[
w_i^{(k)}
=
\frac{p_i^k}{\sum_j p_j^k}.
\]
These weights identify which states are responsible for the collision statistic \(C_k(P)\). 
At low orders, many states can contribute appreciably. 
As \(k\) increases, the powers \(p_i^k\) amplify differences in state probability, shifting the collision statistic toward the highest-probability states.

Figure~\ref{fig:collision-resolution} summarizes this shift using the effective number of contributing states,
\[
N_{\mathrm{eff}}^{(k)}
=
\frac{1}{\sum_i \left(w_i^{(k)}\right)^2}.
\]
This quantity is large when many states contribute comparably to \(C_k(P)\), and small when the statistic is dominated by a few states. 
For the uniform distribution, the effective number remains equal to the full support size because all states contribute equally at every order. 
For the tiered distributions, the effective number decreases rapidly with \(k\). 
The strongest tiered distribution collapses to only a few effective contributors after small increases in collision order, while the milder tiered distributions retain broader support for longer. 
Thus even moderate collision orders act as strong concentration filters.

This result clarifies the interpretation of finite-order collision approximations. 
Adding higher orders does not simply add more information in a neutral way; it changes which parts of the distribution are being probed. 
Low-order collisions retain broader information about the support, while higher-order collisions preferentially emphasize high-mass states and coherent concentration structure. 
The order parameter \(k\) should therefore be understood as a resolution parameter: increasing \(k\) sharpens the view toward dominant events, but can also reduce sensitivity to diffuse or low-probability structure.

\begin{figure}[t]
\centering
\includegraphics[width=\textwidth]{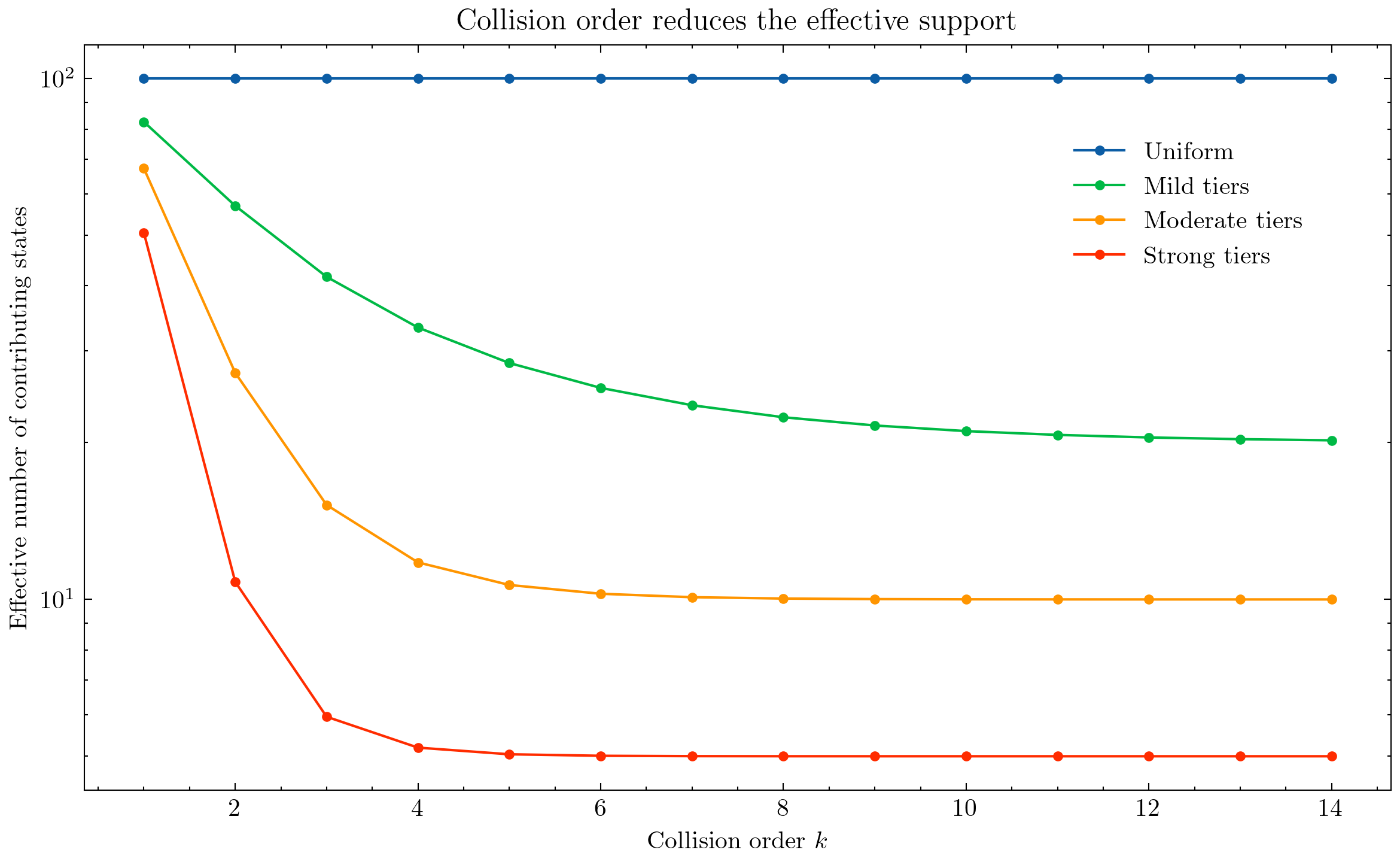}
\caption{
Collision order as a concentration filter. 
For each distribution, the state-level contribution to the order-\(k\) collision probability is normalized as \(w_i^{(k)}=p_i^k/\sum_j p_j^k\). 
The figure shows the effective number of contributing states, \(N_{\mathrm{eff}}^{(k)}=1/\sum_i (w_i^{(k)})^2\), as a function of collision order. 
For the uniform distribution, all states contribute equally and the effective number remains fixed at the support size. 
For increasingly tiered distributions, the effective number decreases with \(k\), showing that higher-order collisions increasingly concentrate on the highest-probability states. 
Thus collision order changes the resolution of the summary: low orders reflect broad support, while higher orders emphasize dominant mass.
}
\label{fig:collision-resolution}
\end{figure}
\unskip

\section{Discussion}
\label{sec:discussion}

This paper develops a finite-resolution view of entropy and mutual information based on low-order collision statistics. 
Rather than treating Shannon entropy as the only population target, we ask what can be recovered from a finite collection of directly estimable coincidence probabilities. 
This leads to a hierarchy of approximations built from integer-order R\'enyi entropies. 
The resulting quantities interpolate between two perspectives: they are grounded in the R\'enyi entropy path, but they remain tied to simple count-based collision events.

The central distinction is between estimation error and approximation error. 
For a fixed collision order, empirical collision probabilities have unbiased fixed-order estimators and converge to their population values as the sample size grows.
Consequently, the corresponding empirical collision entropies converge to fixed R\'enyi entropies, and the extrapolated quantities converge to finite-resolution population targets. 
They do not, in general, converge to Shannon entropy or mutual information. 
The residual gap is deterministic: it reflects the extent to which a finite set of low-order R\'enyi entropies captures the value of the entropy path at $\alpha=1$.

This distinction clarifies the behavior observed in the numerical experiments.
Direct estimators such as NSB are designed to estimate Shannon entropy or mutual information and can perform substantially better in sparse settings \cite{nemenman2002entropy, archer2014bayesian}. 
The collision-based quantities have a different role. 
They expose which parts of Shannon information are already visible from low-order coincidence structure, and they make explicit the residual approximation error induced by finite moment truncation. 
In this sense, they are not competitors to Bayesian entropy estimators so much as finite-resolution summaries with transparent statistical targets.

The collision order also has a structural interpretation. 
As the order increases, collision probabilities become increasingly dominated by the largest probability masses. 
This is the collision-probability counterpart of the convergence of R\'enyi entropy to min-entropy \cite{renyi1961measures, cover2006elements}. 
Thus the order parameter controls the resolution at which a distribution is viewed: low orders aggregate probability mass broadly, while higher orders emphasize repeated observations from high-probability states. 
This provides a formal basis for interpreting low-order collision approximations as scale-dependent or resolution-dependent summaries of uncertainty.

The non-identifiability result places a limit on what such summaries can accomplish.
A finite sequence of collision moments generally does not determine Shannon entropy.
Even if the collision probabilities $C_2(P),\ldots,C_r(P)$ were known exactly, multiple distributions can share those values while having different entropies. 
This shows that the approximation error is not an artifact of a particular estimator. 
It is intrinsic to representing a distribution through finitely many low-order power sums. 
At the same time, the classical relationship between power sums and finite multisets shows that sufficiently many moments can determine a finite distribution up to permutation \cite{macdonald1995symmetric}. 
The framework studied here occupies the intermediate regime: it uses only a small number of collision moments, accepting finite-resolution approximation in exchange for direct estimability and interpretability.

One natural application is the comparison of clusterings. 
Given two partitions of the same elements, the contingency table between cluster labels defines a joint empirical distribution, so standard information-theoretic clustering similarities, including mutual information, normalized mutual information, adjusted mutual information, and variation of information, can be written as functions of this joint distribution and its marginals \cite{meila2007vi, Vinh2010JMLR, gates2017impact}. 
Pair-counting measures such as the Rand index and adjusted Rand index instead depend on whether pairs of elements are co-assigned under the two partitions \cite{Rand1971, HubertArabie1985}. 
Collision probabilities provide a direct bridge between these views: pair-counting corresponds to order-two collision structure, while higher-order collisions measure agreement among triples, quadruples, and larger tuples \cite{gates2025unifying}. 
This connection complements prior work showing that clustering comparisons can be sensitive to random models and structural biases \cite{gates2017impact}, as well as element-centric approaches that compare clusterings through the relationships induced around individual elements \cite{gates2019element}. 
The present paper does not develop a clustering similarity measure directly; instead, it provides the information-theoretic substrate for a related idea: clustering similarities based on pairs, triples, and higher-order tuples can be viewed as finite-resolution points on a collision-information path. 
This suggests a principled way to interpolate between pair-counting and information-theoretic comparisons, while also clarifying the limits of any fixed-order tuple statistic.

The collision perspective also connects naturally to diversity measures in ecology and related fields. 
The order-two collision probability $C_2(P)=\sum_i p_i^2$ is the probability that two independent draws belong to the same category, and its complement is Simpson's diversity index \cite{simpson1949measurement}. 
The reciprocal $1/C_2(P)$ is commonly interpreted as an effective number of types and corresponds to the order-two Hill number \cite{hill1973diversity}. 
Thus, the pairwise collision entropy $H_2(P)$ is closely related to established diversity measures: it is the logarithm of inverse Simpson concentration. 
From this perspective, our framework extends the logic of pairwise diversity from two-sample coincidences to a hierarchy of higher-order coincidence events, and asks how these finite-order diversity summaries approximate Shannon entropy.

Several limitations remain. 
First, the interpolation bounds depend on derivatives of the R\'enyi entropy or contrast path. 
These bounds are useful conceptually, but sharper distribution-dependent bounds would make the approximation theory more predictive. 
Second, higher collision orders introduce a bias--variance tradeoff. 
They may reduce interpolation error by capturing more of the local shape of the R\'enyi path, but they also rely on rarer coincidence events and can be harder to estimate in finite samples. 
Third, the R\'enyi entropy contrast used for mutual information is not a full R\'enyi mutual information in the axiomatic or operational sense. 
Other generalizations, such as Sibson and Arimoto mutual informations, satisfy different properties and serve different purposes \cite{sibson1969information, arimoto1977information, vanerven2014renyi}. 
Here, the contrast is used because it arises naturally from combining marginal and joint collision entropies and because it has the Shannon mutual information as its order-one limit.

Future work could proceed in several directions. 
One direction is to develop sharper approximation bounds for specific distribution classes, such as heavy-tailed distributions, near-uniform distributions, or structured joint distributions. 
A second direction is to study adaptive schemes that choose collision orders based on sample size and observed concentration. 
A third direction is to connect the finite-resolution collision framework to hypothesis testing, where low-order collision contrasts may provide simple statistics for detecting particular forms of dependence. 
Finally, the clustering setting suggests a broader program of tuple-order similarity measures that explicitly control the resolution at which two partitions are compared.

Overall, the collision perspective provides a simple way to separate what is statistically estimable from what is informationally identifiable. 
Low-order collision statistics are not sufficient to determine Shannon entropy in general, nor are they intended to replace direct estimators of Shannon information. 
Their value is that they define a transparent hierarchy of finite-resolution information summaries. 
This hierarchy links coincidence counts, R\'enyi entropy, Shannon limits, and mutual information within a single approximation--estimation framework.

\vspace{6pt} 





\authorcontributions{A.J.G. conceived the study, developed the theoretical framework, derived the main results, designed the numerical experiments, interpreted the results, and wrote the manuscript.}

\funding{The author was supported in part by the National Security Data \& Policy Institute, Contracting Activity \#2024-24070100001.}

\dataavailability{No empirical datasets were analyzed. 
The numerical experiments use synthetic distributions generated from the specifications described in the text. 
Code for reproducing the figures is available at \href{https://github.com/ConnectedDataHub/collision-information}{https://github.com/ConnectedDataHub/collision-information}.}

\acknowledgments{The author would like to thank great conversations with students in his research group, the Connected Data Hub. }

\conflictsofinterest{The authors declare no conflicts of interest.}

\isPreprints{}{
\begin{adjustwidth}{-\extralength}{0cm}
} 

\reftitle{References}


\bibliography{collisionapproxref}

%


\isPreprints{}{
\end{adjustwidth}
} 
\end{document}